\documentclass[11pt,oneside]{article}
\usepackage{color}
\usepackage{psfrag}
\usepackage{epsfig}
\usepackage{graphicx}
\usepackage{amsmath}
\usepackage{amssymb}
\usepackage{amsthm}
\usepackage{cite}
\usepackage[latin1]{inputenc}
\usepackage{cancel}
\usepackage{color}
\usepackage{caption}

\usepackage{float}
\restylefloat{table}

\newtheorem{theorem}{Theorem}[section]
\newtheorem{definition}[theorem]{Definition}

\newtheorem{proposition}[theorem]{Proposition}

\newtheorem{corollary}[theorem]{Corollary}

\begin{document}

\title{PIR codes from combinatorial structures}
\date{}
\author{Massimo Giulietti, Arianna Sabatini, and Marco Timpanella}

\maketitle

\begin{abstract}
A $k$-server Private Information Retrieval (PIR) code is a binary linear $[m,s]$-code admitting a generator matrix such that for every integer $i$ with $1\le i\le s$ there exist $k$ disjoint subsets of columns (called recovery sets) that add up to the vector of weight one, with the single $1$ in position $i$. As shown in \cite{Fazeli1}, a $k$-server PIR code is useful to reduce the storage overhead of a traditional $k$-server PIR protocol. Finding $k$-server PIR codes with a small blocklength for a given dimension has recently become an important research challenge. In this work, we propose new constructions of PIR codes from combinatorial structures, introducing the notion of $k$-partial packing. Several bounds over the existing literature are improved.
\end{abstract}


\section{Introduction}

A Distributed Storage System (DSSs) consists of a set of hard drives (disks), or nodes, and it is used to store data in a distributed manner. DSSs are an integral part of modern data centers which support large scale computing applications. Reasons why one may want to store data in a distributed manner (rather than on a single disk) include ease of scale and reliability. To achieve reliability, redundancy is needed.  Instead of using replication of the nodes, more advanced coding techniques are implemented because of storage efficiency. 

Fazeli, Vardy and Yaacobi \cite{Fazeli1} proposed the definition of a $k$-server PIR code as an important ingredient in the construction of coded PIR protocols. PIR codes are one of the classes of linear codes that received more attention for their applications to DSSs.
A $k$-server PIR code is a binary linear $[m,s]$-code admitting a generator matrix such that for every integer $i$ with $1\le i\le s$ there exist $k$ disjoint subsets of columns (called recovery sets) that add up to the vector of weight one, with the single $1$ in position $i$. Here $m$ is the total number of bits stored on all the servers and $s$ is the number of bits in the database.
Clearly, for given $k$ and $s$ the optimal $m$ is the minimal one. 
Given $k$ and $s$, let $P(s,k)$ denote the least integer $m$ for which a $k$-server PIR $[m,s]$-code exists. 
The storage overhead of a $k$-server PIR $[m,s]$-code is the ratio $m/s$.

Already in \cite{Fazeli1} is was noted that notions and tools from incidence geometry and design theory could be useful to construct good PIR codes. In particular, Lemma 7 in \cite{Fazeli1} states that a collection $S_1,\ldots,S_r$ of subsets of a finite set $X$ such that every element of $X$ belongs to at least $k-1$ subsets and two distinct subsets meet in at most one element give rise to a  $k$-server $[r+s,s]$-code. This result motivates the following definition.

\begin{definition}\label{Def:k-partial} Let $X$ be a finite set of size $s$. A $k$-partial packing of $X$ is a set of $k-1$ partitions of $X$ such that 
\begin{itemize}
\item[(i)] each subset in any partition has size at least two;
\item[(ii)] two subsets from two distinct partitions meet in at most one point. 
\end{itemize}
The {\em order} $r$ of a $k$-partial packing is the total number of subsets of $X$ belonging to its partitions. A $k$-partial packing is {\em homogeneous} if all the subsets from any partition have the same size.
\end{definition}

It is clear that any $k$-partial packing of order $r$ of a set of size $s$ gives rise to a $k$-server PIR $[r+s,s]$-code, with storage overhead $1+\frac{r}{s}$.

A $k$-partial packing $\mathfrak P$ of a set $X$ clearly defines  $h$-partial packings of $X$ for every $h<k$. We will call them {\em partial subpackings of }$\mathfrak P$. It is known that from a $k_1$-server PIR $[m_1,s]$-code and a $k_2$-server PIR $[m_2,s]$-code
one can construct a $(k_1+k_2)$-server PIR $[m_1+m_2,s]$-code; see e.g. \cite[Theorem 2]{Kurz}. Here, it is interesting to note that if $h_1+h_2\le k+1$, then we can construct partial subpackings of $\mathfrak P$ giving rise to an $h_1$-server PIR $[m_1,s]$-code, an
$h_2$-server PIR $[m_2,s]$-code, and an $(h_1+h_2-1)$-server PIR $[m_3,s]$-code with
$$
m_3=m_1+m_2-s<m_1+m_2.
$$
This provides a strong motivation for searching $k$-partial packings with large $k$ with respect to $s$.

Other combinatorial objects which provide $k$-server PIR codes are the so-called configurations; see \cite[Chapter VI, Section 7]{Handbook}.
\begin{definition}\label{Def1-Configur} 
\begin{description}
  \item[i)] A $(v_{t},b_{z})$-configuration is an incidence
      structure of $v$ points and $b$ lines, such that each
      line contains $z$ points, each point lies on $t$
      lines, and  two distinct points are connected by
      {at most} one line.
  \item[ii)] If $v=b$, and hence $t=z$, the configuration
      is \emph{symmetric}, and it is denoted by $v_{z}$.
\end{description}
\end{definition}

It is straightforward to check that a $(v_{t},b_{z})$-configuration produces a $(t+1)$-server PIR code with $s=v$ and storage overhead $1+\frac{b}{v}$. In particular, any symmetric configuration defines a PIR code with storage overhead equal to $2$.
The dual incidence structure of a configuration is still a configuration, which defines a $(z+1)$-server PIR code with $s=b$ and storage overhead $1+\frac{v}{b}$.

We remark that a homogeneous $k$-partial packing of a set $X$, together with its partial subpackings, naturally define configurations.

The aim of this paper is to obtain new upper bounds on $P(s,k)$ through the notions of $k$-partial packings and configurations. Our constructions provide both families of PIR codes whose storage overhead is asymptotically optimal (see Table \ref{table1}), and PIR codes that provide improvements over the existing literature for small values of $s$ and $k$ (see Table \ref{table2}).

We also recall that the PIR codes obtained in this paper are systematic.
Then, by \cite[Corollary 1]{Skachek}, they also produce locally recoverable codes with locality equal to the maximum size of a recovery set and availability $k-1$.

\section{Families of $k$-partial packings}

\subsection{Direct product construction}
Assume that $s$ can be written as the product of $k-1$ integers greater than $2$, that is
$$
s=a_1\cdot a_2\cdots a_{k-1}, \qquad \text{ with } a_i\ge 2.
$$

For an integer $a\ge 2$, let $C_a$ denote the cyclic group of order $a$.
Let 
$$
G=C_{a_1}\times C_{a_2}\times \cdots \times C_{a_{k-1}}
$$
be the direct product of the groups $C_{a_i}$ for $i=1,\ldots,k-1$.

Finally, let $\mathcal P_i$ be the partition induced by the cosets of the subgroup $C_{a_i}$, naturally embedded in $G$.

\begin{proposition} For each $w\le k-1$,
$$
\mathfrak P=\{\mathcal P_1,\ldots, \mathcal P_{w}\}
$$
is a $(w+1)$-partial packing of $G$ of order $\frac{s}{a_1}+\ldots+\frac{s}{a_w}$.
\end{proposition} 
\begin{proof}
As $a_i\geq 2$ for $i=1,\ldots,w$, property (i) of Definition \ref{Def:k-partial} holds. Also, for any two distinct indices $i,j$, the intersection of a coset in $\mathcal{P}_i$ and a coset in $\mathcal{P}_j$ clearly contains at most one element, and hence (ii) holds. Finally, observe that $|\mathcal{P}_i|=\frac{s}{a_i}$ for any $i=1,\ldots,w$.
\end{proof}

The following result is a straightforward corollary. 
\begin{theorem}\label{groupspackings} Let $$
s=a_1\cdot a_2\cdots a_{k-1}, \qquad \text{ with } a_i\ge 2.
$$
Then for each $w\le k-1$ there exists a $(w+1)$-server PIR $[m,s]$-code with
$$
m=s+\frac{s}{a_1}+\ldots+\frac{s}{a_{w}}
$$
and storage overhead $1+\sum_{i=1}^{w}\frac{1}{a_i}$. In particular,  if $s=h^{k-1}$, for each $w\le k-1$ there exists a $(w+1)$-server PIR $[s+w\frac{s}{h},s]$-code with storage overhead $1+\frac{w}{h}$.
\end{theorem}


\subsection{Homogeneous partial packings from  Projective Geometry}

For $q$ a prime power, let ${\rm{PG}}(N,q)$ be the projective space of dimension $N$ over the finite field with $q$ elements $\mathbb F_q$. We recall that the size of ${\rm{PG}}(N,q)$ is
$$
s(N,q)=\frac{q^{N+1}-1}{q-1}=q^N+q^{N-1}+\ldots+q+1,
$$
and the total number of lines is
$$
L(N,q)=\frac{(q^{N+1}-1)(q^N-1)}{(q^2-1)(q-1)}.
$$
Also, a line in ${\rm{PG}}(N,q)$ consists of $q+1$ points, and two distinct lines meet in at most one point.

A {\em resolution class} of ${\rm{PG}}(N,q)$ is a set of lines which partition the point set. A {\em packing} (or {\em resolution}) of the lines of ${\rm{PG}}(N,q)$ is a partition of the lines into resolution classes. Clearly, any $k-1$ resolution classes from a packing are a $k$-partial packing of ${\rm{PG}}(N,q)$.

Sufficient conditions on $N$ and $q$ for a packing to exist
are known since the seventies.
\begin{proposition}\cite{Baker,Beutelspacher}\label{resolution} A packing of the lines of ${\rm{PG}}(N,q)$ exists if
\begin{itemize}
\item[(a)] $N=2z+1$, $q=2$, $z\ge 1$;
\item[(b)] $N=2^{i+1}-1$, $i\ge 1$, $q$ a prime power. 
\end{itemize}
\end{proposition}

Then the following holds.

\begin{theorem} Let $N$ and $q$ be as in (a) or (b) of Proposition \ref{resolution}. Then for $s=s(N,q)$ and any $k\le 1+( q^{N-1}+\ldots+q+1)$, there exists a $k$-server PIR $[m,s]$-code with
$$
m=s+\frac{(k-1)s}{q+1}
$$
and storage overhead $1+\frac{k-1}{q+1}$.
\end{theorem}
\begin{proof}
Note that there are $\ell(N,q)=\frac{s}{q+1}$ lines in any resolution class of ${\rm{PG}}(N,q)$, and a packing of the lines of ${\rm{PG}}(N,q)$ comprises $\frac{L(N,q)}{\ell(N,q)}=\frac{q^N-1}{q-1}=q^{N-1}+\ldots+q+1$ resolution classes. Then the $k$-partial packing of ${\rm{PG}}(N,q)$ obtained taking any $k-1$ resolution classes gives rise to a $k$-server PIR as in the claim.
\end{proof}

\subsection{Homogeneous partial packings from Affine Geometry}\label{2.3}

In ${\rm{AG}}(N,q)$ a resolution is easily obtained for any $N$ and $q$. Here a resolution class is just a parallelism class. Taking into account that every line contains $q$ points, and that the number of parallelism classes is $s(N-1,q)$, the following result is easily obtained.

\begin{theorem}\label{th2.5} Let $q$ be a prime power and $N$ an integer with $N\ge 2$. Then for $s=q^N$ and any $k\le 1+s(N-1,q)$ there exists
a $k$-server PIR $[m,s]$-code with
$$
m=s+\frac{(k-1)s}{q}
$$
and storage overhead $1+\frac{k-1}{q}$.
\end{theorem}

Now we consider subsets $E$ of ${\rm{AG}}(N,q)$ of size $hq^{N-1}$ consisting of $h\le q$ parallel hyperplanes. There are $q^{N-1}$ directions not determined by these hyperplanes and each line with such directions  meets $E$ in precisely $h$ points. Then the following holds.  

\begin{theorem}\label{th2.6} Let $q$ be a prime power and $N$ an integer with $N\ge 2$. Then for $s=hq^{N-1}$, $h\le q$, and any $k\le 1+ q^{N-1}$ there exists
a $k$-server PIR $[m,s]$-code with
$$
m=(h+k-1)q^{N-1}=s+ (k-1)q^{N-1}
$$
and storage overhead $1+\frac{k-1}{h}$.
\end{theorem}


\subsection{Partial packings from other geometrical objects}

\subsubsection{Maximal arcs}
In a projective plane ${\rm{PG}}(2,q)$, a maximal arc is a set  of $v$ points $\mathcal K$ such that every line of ${\rm{PG}}(2,q)$ is either disjoint from $\mathcal K$ or meets $\mathcal K$ in the same number $z$ of points. If this happens $\mathcal K$ is said to be a $\{v; z\}$-maximal arc.

The existence problem for maximal arcs of given size is completely solved; see \cite{SurveyGMT}.
\begin{theorem} A $\{v;z\}$-maximal arc of ${\rm{PG}}(2,q)$ exists if and only if there exist $0\le n'\le n$ such that
$$
q=2^{n}, \qquad z=2^{n'}, \qquad v=zq-q+z.
$$
\end{theorem}

For a point $P$ not in $\mathcal K$, the lines through $P$ that are not disjoint from $\mathcal K$ give rise to a partition of $\mathcal K$ in subsets of size $z$. Also, joining $k-1$ partitions corresponding to $q+1$ collinear points gives rise to a $k$-partial packing of $\mathcal K$. Then the following holds.

\begin{corollary} Let $s$ be an integer of the form
$s=2^{n+n'}-2^{n}+2^{n'}$, for some $1\le n'\le n$. Then for each $k\le 2^n+2$ 
there exists a $k$-server PIR $[m,s]$-code with
$$
m=s+\frac{(k-1)s}{2^{n'}}
$$
and storage overhead $1+\frac{(k-1)}{2^{n'}}$. 
\end{corollary}

\subsubsection{Classical unitals}

A classical unital $U$ in ${\rm{PG}}(2,q^2)$  is the set of points whose homogeneous coordinates $(x_0,x_1,x_2)$ satisfy the equation 
$x_0^{q+1}+x_1^{q+1}+x_2^{q+1}=0$, up to projectivities.

\begin{theorem}\cite[Section 7.3]{Hirschfeld}
The set $U$ consists of $q^3+1$ points, and each point in ${\rm{PG}}(2,q^2)\setminus U$ defines a partition of $U$ in $q^2-q+1$ subsets of $q+1$ collinear points.
\end{theorem}

If we consider a line $l$ meeting $U$ in precisely one point $P$, then the $q^2$ points on $l$ distinct from $P$ define disjoint partitions. Then the following holds.

\begin{corollary} Let $s$ be an integer of the form
$s=q^3+1$, for some prime power $q$. Then for each \textcolor{black}{$k\le q^2+1$} there exists a $k$-server PIR $[m,s]$-code with
$$
m=s+\frac{(k-1)s}{q+1}
$$
and storage overhead $1+\frac{(k-1)}{q+1}$. 
\end{corollary}

\subsubsection{Internal points to a conic}

Let $\mathcal C$ be an irreducible conic in ${\rm{PG}}(2,q)$, with $q$ an odd prime power. A point $P\in {\rm{PG}}(2,q)\setminus \mathcal C$ is external if it lies on a tangent line to $C$, and internal otherwise.

There exist precisely $(q^2-q)/2$ internal points. Also, a secant line of ${\rm{PG}}(2,q)$ contains $(q-1)/2$ points of $\mathcal C$, while an external line contains $(q+1)/2$ internal points of $\mathcal C$. Then clearly the lines through an external point $P$, distinct from the tangent lines at $P$, define a partition of the set of internal points of $\mathcal{C}$ in $q-1$ subsets of collinear points of cardinalities $(q-1)/2$ and $(q+1)/2$; see \cite{GiuliettiLinePartitions}.

If $q>3$, taking $k-1$ distinct external points lying on a same tangent line to $\mathcal{C}$, we obtain a $k$-partial packing of the set of internal points of $\mathcal{C}$. Note that, unlike the other partial packings from geometrical objects, this construction provides a non-homogeneous $k$-partial packings.
The following result then holds.

\begin{corollary} Let $s$ be an integer of the form
$s=(q^2-q)/2$, for some odd prime power $q>3$. Then for each $k\le q+1$ there exists a $k$-server PIR $[m,s]$-code with
$$
m=s+{(k-1)(q-1)},
$$
and storage overhead $1+\frac{2(k-1)}{q}$. 
\end{corollary}


\section{Homogeneous partial packings from resolvable configurations and BIBDs}\label{sec:RBIBD}

Recently, in \cite{Gevay}, the notion of resolvable configuration has been introduced. A parallel class in a configuration $\mathcal{C}$ is a set of lines which partition the set of points; a {\em resolution} of $\mathcal C$ is a partition of the set of lines into parallel classes. A configuration $\mathcal C$ is said to be {\em resolvable} if it admits a resolution. A resolution of a $(v_t,b_z)$ resolvable configuration consists of $t$ parallel classes, each of which has size $\frac{v}{z}$. Therefore, if a $(v_t,b_z)$-configuration is  resolvable, then a $k$-partial packing of the set of its $v$ points can be defined for each $k\le 1+t$.

\begin{theorem}\label{th:ResolvableConf}
Let $(v_t,b_z)$ be a resolvable configuration. Then for any $k\leq 1+t$ there exists a $k$-server PIR $[m,v]$-code with
$$
m=v+(k-1)\frac{v}{z},
$$
and storage overhead $1+\frac{k-1}{z}$.
\end{theorem}

Existence results for symmetric resolvable configurations were investigated in \cite{Buratti1}. Here we list the parameters for which a $v_z$ resolvable configuration exists.
\begin{itemize}
    \item $3\leq z \leq 5$, $v=wz$, $w\geq z$, see \cite[Theorem 3.2]{Buratti1};
    \item $6\leq z \leq 13$, $v=wz$, $w\geq z$, with the following possible exceptions
    $$
    (z,w)\in \{(9, 10),(10, 12),(11, 12),(11, 14),(12, 12),(12, 14),(12, 15),(13, 14),(13, 15)\},
    $$
    see \cite[Theorem 4.7]{Buratti1};
    \item $z\geq 3$, $w \geq z^2$, $v=wz$, see \cite[Corollary 4.6]{Buratti1};
    \item $q$ a prime power, $z\leq q$, $v=zq$, see \cite[Corollary 3.4]{Buratti1}.
\end{itemize}

If a $(v_t,b_z)$-configuration is such that any two distinct points are connected by {\em exactly} one line, then $\mathcal C$ is called a Balanced Incomplete Block Design (BIBD), or a Steiner system. 
In \cite{Fazeli1} it was noticed that one can construct a PIR code from a given Steiner system; see also \cite{RosnesLin}. Here we focus on resolvable Steiner systems, since they give rise to homogeneous partial packings and hence to a large number of distinct PIR codes, each one with a different number of servers. By a counting argument it is easy to see that the number of parallel classes in a resolution of a BIBD is $\frac{v-1}{z-1}$. Therefore, the following results holds.

\begin{theorem}\label{th:RBIBD}
Let $(v_t,b_z)$ be a resolvable BIBD. Then for any $k\leq 1+\frac{v-1}{z-1}$ there exists a $k$-server PIR $[m,v]$-code with
$$
m=v+(k-1)\frac{v}{z}.
$$
\end{theorem}
 
We list here some families of parameters for which there exists a 
$(v_t,b_z)$-configuration which is also a resolvable BIBD; see \cite[Chapter II, Section 7]{Handbook} and \cite{Groppnonsym}.

\begin{itemize}
    \item $z=3$, $v$ such that  $v\equiv 3 \pmod 6$;
    
    \item $z=4$, $v$ such that $v\equiv 4 \pmod {12}$;
    
    \item $z=5$, $v\equiv 5 \pmod {20}$, $v\neq 45,345,465,645$;
    
    \item $z=7$, $v\equiv 7 \pmod {42}$, $v> 294427$;
    
    \item $z=8$, $v \equiv 8 \pmod {56}$, $v>24480$.

\end{itemize}


Also, the following general result holds.
\begin{theorem}\cite[Chapter II, Theorem 7.10]{Handbook}
If $v$ and $z$ are both powers of the same prime, and $z-1$ divides $v-1$, then a $(v_t,b_z)$ resolvable BIBD exists.
\end{theorem}

\section{Families of configurations}

\subsection{Symmetric configurations}\label{Subsec:SymConf}
As already pointed out, any symmetric configuration $v_z$ defines a $(z+1)$-server PIR $[2v,v]$-code with storage overhead equal to $2$. In this section we provide a list of infinite families of symmetric configurations that are known to exist, see \cite{DFGMP,Handbook}. In the following, $q$ is a prime power and $p$ is any prime number.
\begin{center}
    \begin{tabular}{|c|c|c|}\hline
    $v$     &   $z$ &   Conditions \\
    \hline 
    $v$     &   $4$ &   $v\geq 13$\\  \hline 
    $q^2-1$     &   $q$ &   none\\  \hline 
    $p^2-p$     &   $p-1$ &   none\\  \hline 
    $q^2-qs$     &   $q-s$ &   $q>s\geq 0$\\  \hline 
    $q^2-(q-1)s-1$     &   $q-s$ &   $q>s\geq 0$\\ \hline $c(q+\sqrt{q}+1)$     &   $\sqrt{q}+c$ &   $q$ square, $c=2,3,\ldots,q-\sqrt{q}$\\  \hline 
    $2p^2$     &   $p+s$ &   $p+s>0$, $0<s\leq q+1$, $q^2+q+1\leq p$\\  \hline 
    $c(q-1)$     &   $c-\delta$ &  $\delta\geq 0$, $c=\delta,\ldots,b$, $b=q$ if $\delta\geq 1$, $b=\lceil \frac{q}{2}\rceil $ if $\delta=0$ \\ \hline
    $\frac{q(q-1)}{2}$     &   $\frac{q+1}{2}$ &   $q$ odd\\  \hline 
    $\frac{q(q+1)}{2}$     &   $\frac{q-1}{2}$ &   $q$ odd\\  \hline 
    $q^2+q-q\sqrt{q}$     &   $q-\sqrt{q}$ &   $q$ square\\  \hline 
    $q^2-rq-1$     &   $q-r$ &   $q-3\geq r\geq 0$\\ \hline
    $q^2-q-2$     &   $q-1$ &   $q-3\geq r\geq 0$\\ \hline
    $rq-1$     &   $r$ &   $r>0$, $q>r\geq 3$\\ \hline
    $rq-2$     &   $r$ &   $r>0$, $q>r\geq 3$\\ \hline
         \end{tabular}
\end{center}
For small values of $v$ and $z$, more symmetric configurations are known; see \cite[Table 7.13]{Handbook}.
\begin{itemize}
\item    $v\in \{21,23,24,25,26,27,28\}$ and $z=5$;
\item    $v\in \{31,34,35,36,37,38\}$ and $z=6$; 
\item    $v\in \{45,48,49,50\}$     and   $z=7$;
\item    $v\in \{57,63,64\}$     and   $z=8$;
\item    $v\in \{73,78,80\}$     and   $z=9$;
\item   $v\in \{91,98\}$   and  $z=10$;  
\item    $v\in\{133,135\}$  and  $z=12$.
\end{itemize}

\subsection{Non-symmetric configurations}
Non-symmetric configurations allow to obtain PIR codes with storage overhead smaller than $2$. Indeed, let $(v_t,b_z)$ be a configuration with $v\neq b$. Then, up to taking the dual configuration, we can assume $b<v$ and hence this configuration produces a $(t+1)$-server PIR $[v+b,v]$-code, with storage overhead $1+\frac{b}{v}<2$. The existence problem of configurations with $z=3$ is completely solved; see \cite[Theorem 3.1]{Groppnonsym}.
\begin{theorem}\label{thmasym3}
A $(v_t,b_3)$ configuration exists if and only if $vt=3b$ and $v\geq 2t+1$.
\end{theorem}

For $z=4,5$, the following results hold; see \cite[Sections 3.2 and 3.4]{Groppnonsym}.
\begin{theorem}\label{thmasym4}
In the following cases, a configuration $(v_t,b_4)$ exists.
\begin{itemize}
\item $v\equiv 4\pmod{12}$, $v>3t+1$ and $vt=4b$;
\item $v\equiv 0\pmod{12}$, $v\geq 3t+1$, $vt=4b$, and $v\not\in E$, where
$$E=\{84,120,132,180,216,264,312,324,372,456,552,648,660,804,852,888\}
;$$
\item $v\equiv 0\pmod{12}$, $v=3t+3$ and $vt=4b$;
\item $t=4s$, $v\geq 3t+1$, $vt=4b$, and $1\leq s\leq 15$, except possibly $s=3$ and $v=38$;
\item $t=6$, $v\geq 20$ even, $b=\frac{3v}{2}$.
\end{itemize}
\end{theorem}

\begin{theorem}\label{thmasym5}
In the following cases, a configuration $(v_t,b_5)$ exists.
\begin{itemize}
\item $v=4t+4$, $v\equiv 0\pmod{20}$, and $vt=5b$;
\item $v\equiv 5\pmod{20}$, $v\geq 4t+1$, $vt=5b$, and $v\geq 7865$;
\item $t=5s$, $v\geq 4t+1$, $vt=5b$, and $1\leq s\leq 10$, except possibly for the cases $(t,v)\in E$, where
\begin{eqnarray*}
E&=&\{(1,22),(2,42), (2,43), (3,62) (3,63) (4,82), (5,102), (7,142) (9,182), (9,183), (9,185), \\
&&(9,186), (9,187), (9,188), (9,189), (9,190), (9,191), (9,192) \}
.
\end{eqnarray*}
\end{itemize}
\end{theorem}

\subsection{Asymptotic results}

It was proven in \cite{bras} that for fixed $t$ and $z$ there exist integers $v_0,\, b_0$ such that for every 
$v\ge v_0$ and $b\ge b_0$ with $vt=bz$,
there exists a $(v_t,b_z)$-configuration.

This means that if we fix the number of server $k$ and an arbitrary fraction of  $\frac{r}{k+1}$ then for $s$ sufficiently large and such that $sr$ is a multiple of $k+1$,  there exists a $k$-server PIR $[m,s]$-code with $m=s(1+\frac{r}{k+1})$ and storage overhead $(1+\frac{r}{k+1})$.

\subsection{Dual configurations from partial packings}\label{sec:dualfrompack}

In the direct product construction, if $G=C_h^\ell$ we obtain an homogeneous partial packing. Since it defines a configuration, we can also consider the dual configuration. Therefore, $k$-server PIR $[m,s]$-codes with the following parameters are obtained:
\begin{itemize}
\item any $h$, any $\ell$:
$$
k=h+1, \qquad s=vh^{\ell-1} \text{ with } 2\le v\le \ell,  \qquad m=s+h^\ell,
$$
and storage overhead $1+\frac{h}{v}$.
\end{itemize}


The same approach can be used for the other constructions that provide homogeneous partial packings. Therefore, we obtain PIR codes with the following parameters:
\begin{itemize}
\item Projective case ($q$ and $N$ as in (a) or (b) of Proposition \ref{resolution}):
$$
k=q+2,\qquad s=v\frac{s(N,q)}{q+1} \text{ with } 2\le v\le (q+1)\frac{L(N,q)}{s(N,q)}, \qquad m=s+s(N,q),
$$
and storage overhead $1+\frac{q+1}{v}$.


\item Affine case, Theorem \ref{th2.6} (any $q$ prime power, $N\ge 2$):
$$
k\le q+1,\qquad s=vq^{N-1} \text{ with } 2\le v\le q^{N-1},\qquad m=s+(k-1)q^{N-1},
$$
and storage overhead $1+\frac{k-1}{v}$.

\item Maximal arcs case (maximal arcs of size $2^{n+n'}-2^{n}+2^{n'}$, for some $0\le n'\le n$):
$$
k=2^{n'}+1,\qquad s=h({2^{n}-2^{n-n'}+1}) \text{ with }
2\le h \le 2^n+1, \qquad 
m=s+2^{n+n'}-2^n+2^{n'}
$$
and storage overhead $1+\frac{2^{n'}}{h}$.

\item Classical unitals case: 
$$
k=q+2,\qquad s=h(q^2-q+1) \text{ with }
2\le h\le q^2, \qquad m=s+q^3+1
$$
and storage overhead $1+\frac{q+1}{h}$.

\item Resolvable BIBD case: if a $(v_t,b_z)$-configuration which is also a resolvable BIBD exists, then the dual construction provide $k$-server PIR $[m,s]$-codes with
$$
k=z+1,\qquad s=h\frac{v}{z} \text{ with }
h\le\frac{v-1}{z-1},\qquad m=s+v 
$$
and storage overhead $1+\frac{z}{h}$.


\end{itemize}

\section{General constructions of $k$-server PIR codes}
In the previous sections we constructed PIR codes whose lengths had a specific form. Here we explicitly construct PIR codes of arbitrary length.

The proof of the following statement is straightforward.
\begin{proposition}\label{reduction} Let $\mathfrak{P}=\{\mathcal P_1,\ldots, \mathcal P_{k-1}\}$ be a $k$-partial packing of a set $X$. Let $Y$ be a subset of $X$ and for each $i=1,\ldots,k-1$ let $\mathcal P_i^Y$ be the partition of $Y$ induced by $\mathcal P_i$. Then
$\mathfrak{P}^Y=\{\mathcal P_1^Y,\ldots, \mathcal P_{k-1}^Y\}$ is a $k$-partial packing of $Y$ if and only if for each $i$ no subset of $\mathcal P_i$ meets $Y$ in precisely one element. In this case, the order of $\mathfrak P^Y$ is less than or equal to that of $\mathfrak P$. 
\end{proposition}


As an illustration, we apply Proposition \ref{reduction} to the partial packings described in Section \ref{2.3}.

Let $q^N$ be the least prime power such that $k\le 1+s(N-1,q)-q^{N-1}$ and $s\le q^N$. The condition on $k$ allows to  construct a $k$-partial packing $\mathfrak{P}$ according to Theorem \ref{th2.5}, in which the parallelism classes of the lines belonging to a fixed hyperplane $H$ are avoided.

If in addition $s\ge 2q^{N-1}$, then one can fix a subset $Y$ of ${\rm{AG}}(N,q)$ with size $s$ that contains two hyperplanes parallel to $H$. Then clearly every line belonging to the partitions of $\mathfrak P$ meets $Y$ in at least two points, and $\mathfrak{P}^Y$ is a $k$-partial packing.

\begin{theorem} For integers $k$ and $s$, let $q^N$ be the least prime power such that $k\le 1+s(N-1,q)-q^{N-1}$ and $2q^{N-1}\le s\le q^N$. Then there exists a $k$-server PIR $[m,s]$-code with
$$
m=s+(k-1)q^{N-1}
$$
and storage overhead $1+\frac{(k-1)q^{N-1}}{s}$. 
\end{theorem}
The best case is clearly when $s$ is close to a prime power. However, something very general can be stated.
\begin{corollary}
For integers $k$ and $s$, let $q^N$ be the least prime power such that $k\le 1+s(N-1,q)-q^{N-1}$ and $2q^{N-1}\le s\le q^N$. Then there exists a $k$-server PIR $[m,s]$-code with storage overhead $O$ with
$$
1+\frac{k-1}{q}\le O \le  1+\frac{k-1}{2}.
$$

\end{corollary}



\section{Conclusions}
In recent years, finding $k$-server PIR codes with a small blocklength for a given dimension has become an important research challenge.
Let $P(s,k)$ denote the minimum value of $m$ for which a $k$-server PIR $[m,s]$-code exists.

In this paper several upper bounds on $P(s,k)$ have been obtained through the notions of $k$-partial packings and configurations. Here we summarize our result on $P(s,k)$, taking into account that the function $P$ is strictly increasing in both variables $s$ and $k$, as the following propagation rules show.
\begin{proposition}\cite[Lemmas 13 and 14]{Fazeli2}
\begin{itemize}
    \item[(i)] $P(s,k)\leq P(s,k+1)-1$;
    \item[(ii)] if $k$ is odd, then $P(s,k)= P(s,k+1)-1$;
    \item[(iii)] $P(s,k)\leq P(s+1,k)-1$.
\end{itemize}
\end{proposition}

In the following table $q$ denotes a prime power, whereas $N$ and $a_i$ any  integer greater than $1$. The integer $k$ is always assumed to be greater than $2$.
\begin{table}[H]
\caption{New upper bounds on $P(s,k)$.}\label{table1}
\begin{center}
    \begin{tabular}{c|c|c}
    $s$     & $k$ & $P(s,k)\le$ \\
    \hline \hline 
 $a_1\cdot a_2\cdots a_{c}$ & $\le c+1$ & $s(1+\frac{1}{a_1}+\cdots+\frac{1}{a_{k-1}})$ \\ \hline  
 $2^{N+1}-1$, $N$ odd & $ \le 2^N$ & $s(1+\frac{k-1}{3})$ \\ \hline  
 $\frac{q^{N+1}-1}{q-1}$, $N=2^{i+1}-1$ & $\le 1+\frac{q^N-1}{q-1}$ & $s(1+\frac{k-1}{q+1})$ \\ \hline  
 $q^N$  & $ \le 1+\frac{q^N-1}{q-1}$ & $s(1+\frac{k-1}{q})$ \\ \hline  
  $2^{n+n'}-2^n+2^{n'}$, $0\le n'\le n$  & $\le 2^{n}+2$ & $s(1+\frac{k-1}{2^{n'}})$\\ \hline  
    $q^3+1$ & $ \le q^2+1$ & $s(1+\frac{k-1}{q+1})$\\ \hline  
    $\frac{q^2-q}{2}$  & $ \le q+1$& $ s+(k-1)q$ \\ \hline  
      $\equiv 3 \pmod 6$ & $\le 1+\frac{s-1}{2}$ & $s(1+\frac{k-1}{3})$\\ \hline  
    $\equiv 4 \pmod {12}$    & $\le 1+\frac{s-1}{3}$& $s(1+\frac{k-1}{4})$\\ \hline  
    $\equiv 5 \pmod {20}$, $\neq 45,345,465,645$     &  $\le 1+\frac{s-1}{4}$& $s(1+\frac{k-1}{5})$ \\ \hline  
    $\equiv 7 \pmod {42}$   $>294427$   & $\le 1+\frac{s-1}{6}$& $s(1+\frac{k-1}{7})$\\ \hline  
    $\equiv 8 \pmod {56}$, $>24480$       & $\le 1+\frac{s-1}{7}$& $s(1+\frac{k-1}{8})$\\ \hline  
    $sh$ multiple of $k+1$, sufficiently large       & arbitrary & $s(1+\frac{h}{k+1})$\\ \hline  
$s\ge 13$       & $3$ & $2s$\\ \hline  
 \end{tabular}
\end{center}\end{table}

Finally, in the next table we report the best known bounds for $P(s, k)$ for small values of $s$ and $k$. In particular, the improvements 
over the existing literature that are provided by our constructions are printed in bold. In these cases, we state the Section (briefly S), or Theorem (briefly T) from which the improvement is obtained. Also, we use PR to denote the improvements that are obtained using the constructions of this paper together with the above-mentioned propagation rules.

\begin{table}[H]
\caption{Best known bounds for $P(s, k)$ for small values of $s$ and $k$.}\label{table2}
\begin{center}\scalebox{0.93}{
	\begin{tabular}{c||c|c|c|c|c|c|c|c|c|c|c|c|c|}
			\textbf{k} $\setminus$ \textbf{t} & \multicolumn{2}{c|}{\textbf{2}} & \multicolumn{2}{c|}{\textbf{3}} & \multicolumn{2}{c|}{\textbf{4}} &
			\multicolumn{2}{c|}{\textbf{5}} &
			\multicolumn{2}{c|}{\textbf{6}} &
			\multicolumn{2}{c|}{\textbf{7}} \\
			\hline \hline
			\textbf{2} & $3^*$ & 1.50 & $5^*$ & 2.50 & $6^*$ & 3.00 & $8^*$ & 4.00 & $9^*$ & 4.50 & $11^*$ & 5.50 \\
			\hline
			\textbf{3} & $4^*$ & 1.33 & $6^*$ & 2.00 & $7^*$ & 2.33 & $10^*$ & 3.33 & $11^*$ & 3.67 & $13^*$ & 4.33 \\
			\hline
			\textbf{4} & $5^*$ & 1.25 & 8 & 2.00 & 9 & 2.25 & 11 & 2.75 & $12^*$ & 3.00 & 14 & 3.50 \\
			\hline
			\textbf{5} & $6^*$ & 1.20 & 9 & 1.80 & 10 & 2.00 & 12 & 2.40 & 13 & 2.60 & 17 & 3.40 \\
			\hline
			\textbf{6} & $7^*$ & 1.17 & 10 & 1.67 & 11 & 1.83 & 13 & 2.17 & 14 & 2.33 & 18 & 3.00 \\
			\hline
			\textbf{7} & $8^*$ & 1.14 & 12 & 1.71 & 13 & 1.86 & 14 & 2.00 & 15 & 2.14 & 20 & 2.86 \\
			\hline
			\textbf{8} & $9^*$ & 1.13 & 13 & 1.63 & 14 & 1.75 & 17 & 2.13 & 18 & 2.25 & 22 & 2.75 \\
			\hline
			\textbf{9} & $10^*$ & 1.11 & 14 & 1.56 & 15 & 1.67 & 19 & 2.11 & 20 & 2.22 & 24 & 2.67 \\
			\hline
			\textbf{10} & $11^*$ & 1.10 & 15 & 1.50 & 16 & 1.60 & 20 & 2.00 & 21 & 2.10 & 25 & 2.50 \\
			\hline
			\textbf{11} & $12^*$ & 1.09 & $\textbf{17}^{PR}$ & \textbf{1.55} & $\textbf{18}^{PR}$ & \textbf{1.64} & 24 & 2.18 & 25 & 2.27 & 36 & 3.27 \\
			\hline
			\textbf{12} & $13^*$ & 1.08 & $\textbf{18}^{T.\ref{thmasym4}}$ & \textbf{1.50} & $\textbf{20}^{PR}$ & \textbf{1.67} & 25 & 2.08 & 26 & 2.17 & 38 & 3.17 \\
			\hline
			\textbf{13} & $14^*$ & 1.08 & 21 & 1.62 & 22 & 1.69 & 26 & 2.00 & 27 & 2.08 & $\textbf{39}^{T. \,\ref{thmasym3}}$ & \textbf{3.00}\\
			\hline
			\textbf{14} & $15^*$ & 1.07 & 22 & 1.57 & 23 & 1.64 & 28 & 2.00 & 29 & 2.07 & 42 & 3.00 \\
			\hline
			\textbf{15} & $16^*$ & 1.07 & 23 & 1.53 & 24 & 1.60 & $\textbf{30}^{S.\,\ref{Subsec:SymConf}}$ & \textbf{2.00} & $\textbf{31}^{PR}$ & \textbf{2.07} & 43 & 2.87 \\
			\hline
			\textbf{16} & $17^*$ & 1.06 & 24 & 1.50 & 25 & 1.56 & $\textbf{32}^{T.\, \ref{th:RBIBD}}$ & \textbf{2.00} & $\textbf{33}^{PR}$ & \textbf{2.06} & 44 & 2.75 \\
			\hline
			\textbf{17} & $18^*$ & 1.06 & $\textbf{26}^{PR}$ & \textbf{1.53} & $\textbf{27}^{PR}$ & \textbf{1.59} & $\textbf{33}^{PR}$ & \textbf{1.94} & $\textbf{34}^{PR}$ & \textbf{2.00} & 45 & 2.65 \\
			\hline
			\textbf{18} & $19^*$ & 1.06 & $\textbf{27}^{T.\ref{groupspackings}}$ & \textbf{1.50} & $\textbf{28}^{PR}$ & \textbf{1.56} & $\textbf{34}^{PR}$ & \textbf{1.89} & $\textbf{35}^{PR}$ & \textbf{1.94} & 46 & 2.56 \\
			\hline\textbf{19} & $20^*$ & 1.05 & $\textbf{28}^{PR}$ & \textbf{1.47} & $\textbf{29}^{PR}$ & \textbf{1.53} & $\textbf{35}^{PR}$ & \textbf{1.84} & $\textbf{36}^{PR}$ & \textbf{1.89} & 47 & 2.47 \\
			\hline
			\textbf{20} & $21^*$ & 1.05 & $\textbf{29}^{T.\,\ref{groupspackings}}$ & \textbf{1.45} & $\textbf{30}^{PR}$ & \textbf{1.50} & $\textbf{36}^{T.\,\ref{thmasym5}}$ & \textbf{1.80} & $\textbf{37}^{PR}$ & \textbf{1.85} & 48 & 2.40 \\
			\hline
			\textbf{21} & $22^*$ & 1.05 & 31 & 1.48 & 32 & 1.52 & 41 & 1.95 & 42 & 2.00 & 49 & 2.33 \\
			\hline
			\textbf{22} & $23^*$ & 1.05 & 32 & 1.45 & 33 & 1.50 & $\textbf{42}^{PR}$ & \textbf{1.91} & $\textbf{43}^{PR}$ & \textbf{1.95} & 50 & 2.27 \\
			\hline
			\textbf{23} & $24^*$ & 1.04 & 33 & 1.43 & 34 & 1.48 & $\textbf{43}^{PR}$ & \textbf{1.87} & $\textbf{44}^{PR}$ & \textbf{1.91} & 51 & 2.22 \\
			\hline
			\textbf{24} & $25^*$ & 1.04 & 34 & 1.42 & 35 & 1.46 & $\textbf{44}^{PR}$ & \textbf{1.83} & $\textbf{45}^{PR}$ & \textbf{1.88} & 52 & 2.17 \\
			\hline
			\textbf{25} & $26^*$ & 1.04 & 35 & 1.40 & 36 & 1.44 & $\textbf{45}^{S.\, \ref{sec:RBIBD}}$ & \textbf{1.80} & $\textbf{46}^{PR}$ & \textbf{1.84} & 53 & 2.22 \\
			\hline
			\textbf{26} & $27^*$ & 1.04 & $\textbf{37}^{PR}$ & \textbf{1.42} & $\textbf{38}^{PR}$ & \textbf{1.46} & $\textbf{46}^{PR}$ & \textbf{1.77} & $\textbf{47}^{PR}$ & \textbf{1.81} & 54 & 2.08 \\
			\hline
			\textbf{27} & $28^*$ & 1.04 & $\textbf{38}^{PR}$ & \textbf{1.41} & $\textbf{39}^{PR}$ & \textbf{1.44} & $\textbf{47}^{PR}$ & \textbf{1.74} & $\textbf{48}^{PR}$ & \textbf{1.78} & 55 & 2.04 \\
			\hline
			\textbf{28} & $29^*$ & 1.04 & $\textbf{39}^{T.\,\ref{groupspackings}}$ & \textbf{1.39} & $\textbf{40}^{PR}$ & \textbf{1.43} & $\textbf{48}^{PR}$ & \textbf{1.71} & $\textbf{49}^{PR}$ & \textbf{1.75} & 56 & 2.00 \\
			\hline
			\textbf{29} & $30^*$ & 1.03 & $\textbf{40}^{PR}$ & \textbf{1.38} & $\textbf{41}^{PR}$ & \textbf{1.41} & $\textbf{49}^{PR}$ & \textbf{1.69} & $\textbf{50}^{PR}$ & \textbf{1.72} & 57 & 1.97 \\
			\hline
			\textbf{30} & $31^*$ & 1.03 & $\textbf{41}^{T.\,\ref{groupspackings}}$ & \textbf{1.37} & $\textbf{42}^{PR}$ & \textbf{1.40} & $\textbf{50}^{T. \,\ref{thmasym4}}$ & \textbf{1.67} & $\textbf{51}^{PR}$ & \textbf{1.70} & 58 & 1.93 \\
			\hline
		\end{tabular}}
\end{center}
\end{table}

\section*{Acknowledgments}
This research was partially supported  by the Italian National Group for Algebraic and Geometric Structures and their Applications (GNSAGA - INdAM). The first author is funded by the project ``Strutture Geometriche, Combinatoria e loro Applicazioni'' (Fondo Ricerca di Base, 2019, University of Perugia). The third author is funded by
the project ``VALERE: VAnviteLli pEr la RicErca" of the University of Campania ``Luigi Vanvitelli". 
The authors would like to thank Marco Buratti for his helpful suggestions.

\end{document}